\documentclass{article}

\PassOptionsToPackage{numbers, compress}{natbib}
\usepackage{natbib}
\usepackage[verbose=true,letterpaper]{geometry}

\makeatletter

\AtBeginDocument{
  \newgeometry{
    textheight=9in,
    textwidth=5.5in,
    top=1in,
    headheight=12pt,
    headsep=25pt,
    footskip=30pt
  }
  \@ifpackageloaded{fullpage}
    {\PackageWarning{neurips_2026}{fullpage package not allowed! Overwriting formatting.}}
    {}
}

\renewcommand{\normalsize}{%
  \@setfontsize\normalsize\@xpt\@xipt
  \abovedisplayskip      7\p@ \@plus 2\p@ \@minus 5\p@
  \abovedisplayshortskip \z@ \@plus 3\p@
  \belowdisplayskip      \abovedisplayskip
  \belowdisplayshortskip 4\p@ \@plus 3\p@ \@minus 3\p@
}
\normalsize
\renewcommand{\small}{%
  \@setfontsize\small\@ixpt\@xpt
  \abovedisplayskip      6\p@ \@plus 1.5\p@ \@minus 4\p@
  \abovedisplayshortskip \z@  \@plus 2\p@
  \belowdisplayskip      \abovedisplayskip
  \belowdisplayshortskip 3\p@ \@plus 2\p@   \@minus 2\p@
}

\renewcommand{\footnotesize}{\fontsize{8pt}{9.5pt}\selectfont}
\renewcommand{\large}{\@setfontsize\large\@xiipt{14}}
\renewcommand{\Large}{\@setfontsize\Large\@xivpt{16}}
\renewcommand{\LARGE}{\@setfontsize\LARGE\@xviipt{20}}
\renewcommand{\huge}{\@setfontsize\huge\@xxpt{23}}
\renewcommand{\Huge}{\@setfontsize\Huge\@xxvpt{28}}

\providecommand{\section}{}
\renewcommand{\section}{%
  \@startsection{section}{1}{\z@}%
                {-2.0ex \@plus -0.5ex \@minus -0.2ex}%
                { 1.5ex \@plus  0.3ex \@minus  0.2ex}%
                {\large\bf\raggedright}%
}
\providecommand{\subsection}{}
\renewcommand{\subsection}{%
  \@startsection{subsection}{2}{\z@}%
                {-1.8ex \@plus -0.5ex \@minus -0.2ex}%
                { 0.8ex \@plus  0.2ex}%
                {\normalsize\bf\raggedright}%
}
\providecommand{\subsubsection}{}
\renewcommand{\subsubsection}{%
  \@startsection{subsubsection}{3}{\z@}%
                {-1.5ex \@plus -0.5ex \@minus -0.2ex}%
                { 0.5ex \@plus  0.2ex}%
                {\normalsize\bf\raggedright}%
}
\providecommand{\paragraph}{}
\renewcommand{\paragraph}{%
  \@startsection{paragraph}{4}{\z@}%
                {1.5ex \@plus 0.5ex \@minus 0.2ex}%
                {-1em}%
                {\normalsize\bf}%
}
\providecommand{\subparagraph}{}
\renewcommand{\subparagraph}{%
  \@startsection{subparagraph}{5}{\z@}%
                {1.5ex \@plus 0.5ex \@minus 0.2ex}%
                {-1em}%
                {\normalsize\bf}%
}

\newlength{\@neuripsabovecaptionskip}
\newlength{\@neuripsbelowcaptionskip}
\renewenvironment{table}
  {\setlength{\abovecaptionskip}{\@neuripsbelowcaptionskip}%
   \setlength{\belowcaptionskip}{\@neuripsabovecaptionskip}%
   \@float{table}}
  {\end@float}

\renewcommand{\footnoterule}{\kern-3\p@ \hrule width 12pc \kern 2.6\p@}
\def\@listi{\leftmargin\leftmargini}
\def\@listii{\leftmargin\leftmarginii
              \labelwidth\leftmarginii
              \advance\labelwidth-\labelsep
              \topsep 2\p@ \@plus 1\p@ \@minus 0.5\p@
              \parsep 1\p@ \@plus 0.5\p@ \@minus 0.5\p@
              \itemsep\parsep}
\def\@listiii{\leftmargin\leftmarginiii
               \labelwidth\leftmarginiii
               \advance\labelwidth-\labelsep
               \topsep 1\p@ \@plus 0.5\p@ \@minus 0.5\p@
               \parsep\z@
               \partopsep 0.5\p@ \@plus 0\p@ \@minus 0.5\p@
               \itemsep\topsep}
\def\@listiv{\leftmargin\leftmarginiv
              \labelwidth\leftmarginiv
              \advance\labelwidth-\labelsep}
\def\@listv{\leftmargin\leftmarginv
             \labelwidth\leftmarginv
             \advance\labelwidth-\labelsep}
\def\@listvi{\leftmargin\leftmarginvi
              \labelwidth\leftmarginvi
              \advance\labelwidth-\labelsep}

\providecommand{\maketitle}{}
\renewcommand{\maketitle}{%
  \par
  \begingroup
    \renewcommand{\thefootnote}{\fnsymbol{footnote}}
    \renewcommand{\@makefnmark}{\hbox to \z@{$^{\@thefnmark}$\hss}}
    \long\def\@makefntext##1{%
      \parindent 1em\noindent
      \hbox to 1.8em{\hss $\m@th ^{\@thefnmark}$}##1
    }
    \thispagestyle{empty}
    \@maketitle
    \@thanks
    \@notice
  \endgroup
  \let\maketitle\relax
  \let\thanks\relax
}

\newcommand{\@toptitlebar}{
  \hrule height 4\p@
  \vskip 0.25in
  \vskip -\parskip%
}
\newcommand{\@bottomtitlebar}{
  \vskip 0.29in
  \vskip -\parskip
  \hrule height 1\p@
  \vskip 0.09in%
}

\providecommand{\@maketitle}{}
\renewcommand{\@maketitle}{%
  \vbox{%
    \hsize\textwidth
    \linewidth\hsize
    \vskip 0.1in
    \@toptitlebar
    \centering
    {\LARGE\bf \@title\par}
    \@bottomtitlebar
    \def\And{%
      \end{tabular}\hfil\linebreak[0]\hfil%
      \begin{tabular}[t]{c}\bf\rule{\z@}{24\p@}\ignorespaces%
    }
    \def\AND{%
      \end{tabular}\hfil\linebreak[4]\hfil%
      \begin{tabular}[t]{c}\bf\rule{\z@}{24\p@}\ignorespaces%
    }
    \begin{tabular}[t]{c}\bf\rule{\z@}{24\p@}\@author\end{tabular}%
    \vskip 0.3in \@minus 0.1in
  }
}

\newcommand{\ftype@noticebox}{8}
\newcommand{\@noticestring}{Preprint.}
\newcommand{\@notice}{%
  \enlargethispage{2\baselineskip}%
  \@float{noticebox}[b]%
    \footnotesize\@noticestring%
  \end@float%
}

\renewenvironment{abstract}%
{%
  \vskip 0.075in%
  \centerline%
  {\large\bf Abstract}%
  \vspace{0.5ex}%
  \begin{quote}%
}
{
  \par%
  \end{quote}%
  \vskip 1ex%
}

\makeatother

\usepackage[utf8]{inputenc}
\usepackage[T1]{fontenc}
\usepackage{hyperref}
\usepackage{booktabs}
\usepackage{microtype}

\hypersetup{
  hidelinks,
  hypertexnames=false,
  pdftitle={Quantum Speedups for Log-Concave Sampling from Local Structure},
  pdfauthor={Chenghua Liu, Qisheng Wang, Zhengfeng Ji}
}

\usepackage{amsmath,amsthm,amssymb}
\usepackage{algorithm}
\usepackage{algpseudocode}
\usepackage{cleveref}
\usepackage{float}
\usepackage{bm}
\usepackage{graphicx}

\newtheorem{theorem}{Theorem}
\newtheorem{lemma}[theorem]{Lemma}
\newtheorem{proposition}[theorem]{Proposition}
\newtheorem{corollary}[theorem]{Corollary}
\newtheorem{definition}{Definition}
\newtheorem{remark}[theorem]{Remark}

\AddToHook{env/theorem/begin}{\crefalias{theorem}{theorem}}
\AddToHook{env/lemma/begin}{\crefalias{theorem}{lemma}}
\AddToHook{env/proposition/begin}{\crefalias{theorem}{proposition}}
\AddToHook{env/corollary/begin}{\crefalias{theorem}{corollary}}
\AddToHook{env/remark/begin}{\crefalias{theorem}{remark}}

\newcommand{\microspace}{\mspace{.5mu}}
\newcommand{\paren}[1]{\left( #1 \right)}
\newcommand{\parens}[1]{( #1 )}
\newcommand{\ket}[1]{\ensuremath{\lvert\microspace #1
    \microspace\rangle}}
\newcommand{\bra}[1]{\ensuremath{\langle\microspace #1
    \microspace\rvert}}
\renewcommand{\d}{\operatorname{d}}
\newcommand{\polylog}{\operatorname{polylog}}
\newcommand{\TV}{\operatorname{TV}}
\newcommand{\gap}{\operatorname{gap}}
\renewcommand{\bar}[1]{\overline{#1}}
\def\cO{\mathcal{O}}
\newcommand{\R}{\mathbb{R}}
\newcommand{\E}{\mathbb{E}}

\title{Quantum Speedups for Log-Concave Sampling from Local Structure}

\author{%
  Chenghua Liu\textsuperscript{1,2}
  \qquad
  Qisheng Wang\textsuperscript{3}
  \qquad
  Zhengfeng Ji\textsuperscript{4}\\[2mm]
  \textsuperscript{1}Institute of Software, Chinese Academy of Sciences, Beijing, China\\
  \textsuperscript{2}University of Chinese Academy of Sciences, Beijing, China\\
  \textsuperscript{3}School of Computer Science, Shanghai Jiao Tong University, Shanghai, China\\
  \textsuperscript{4}Department of Computer Science and Technology, Tsinghua University, Beijing, China
}

\begin{document}

\maketitle

\begin{abstract}
For a convex function \(f \colon \mathbb{R}^d \to \mathbb{R}\), the problem of sampling from a distribution proportional to \(e^{-f(x)}\) is called log-concave sampling.
In many practical scenarios, the function $f(x)$ turns out to admit a local decomposition $f(x) = \sum_{a=1}^R \psi_a(x_{S_a})$.
In this paper, we consider log-concave sampling using \textit{local} queries, i.e., evaluation and gradient queries to each clause $\psi_a(\cdot)$, which can be computationally much cheaper than the queries to $f(x)$ itself.
We show that if each coordinate appears in only a small number of clauses, there is a quantum algorithm for strongly log-concave sampling using $\widetilde{O}(\sqrt{\kappa}d)$ local queries, where $\kappa$ is the condition number.
This improves the prior best classical result $\widetilde{O}(\kappa d)$ due to \hyperlink{cite.ALZ24}{Ascolani, Lavenant, and Zanella (\textit{Ann.\ Probab.}\ 2026)} and quantum result $\widetilde{O}(\sqrt{\kappa} d^2)$ implied by \hyperlink{cite.CLL22}{ Childs et al.\ (NeurIPS 2022)}.\footnote{In \cite{CLL22}, their approach uses $\widetilde{O}(\sqrt{\kappa}d)$ \textit{global} queries, i.e., evaluation and gradient queries to $f(x)$, and thus implies an approach using $\widetilde{O}(\sqrt{\kappa}d^2)$ local queries since each global query requires $\Theta(d)$ local queries to implement even if each coordinate appears in $O(1)$ clauses.}
Our quantum sampler applies to a broad class of locally structured models from statistical computing and machine learning, with representative examples including Gaussian Markov random fields, finite-element latent Gaussian models, and sparse generalized linear models.
These results demonstrate that local structure is not merely an implementation detail, but a quantum algorithmic resource for high-dimensional sampling.
\end{abstract}

\section{Introduction}
Sampling from log-concave distributions is a fundamental primitive in statistics, optimization, machine learning, and statistical physics.
In this paper, we study quantum algorithms for sampling from a strongly log-concave target distribution on \(\R^d\), with density \(\pi(x)=Z_f^{-1}e^{-f(x)}\), where \(Z_f\) is the normalizing constant.
We assume that the potential \(f:\R^d\to\R\) is \(L\)-smooth and \(\lambda\)-strongly convex, that is,
\begin{equation}
\label{eq:intro-strongly-convex}
\|\nabla f(x)-\nabla f(y)\|\le L\|x-y\|,
\qquad
f(y)-f(x)\ge \langle \nabla f(x),y-x\rangle+\frac{\lambda}{2}\|y-x\|^2
\end{equation}
for all $x,y\in\R^d$. We write
$\kappa:=\frac{L}{\lambda}$
for the condition number. Equivalently, $\pi$ is a smooth strongly log-concave distribution.

Existing quantum algorithms for continuous sampling typically assume coherent global access to the target potential, specifically via evaluation queries to $f$ or gradient queries to $\nabla f$ \cite{CLL22,CCH23,OLMW24,OLMW25,LDCL26}. The same global query model is also standard in quantum algorithms for continuous optimization \cite{vAGGdW20,ZLL21,LZ22,ZZF24,CLW25,GZL25,LS25}. We refer to these queries as \emph{global} because a single application requires computing the complete scalar potential $f(x)$ or the full gradient vector $\nabla f(x)$, both fundamentally dependent on the entire $d$-dimensional configuration.

\begin{definition}[Quantum global query]
For a potential \(f:\R^d\to\R\), the \emph{quantum global evaluation oracle} is the unitary
\[
\cO_f\ket{x,z}
=
\ket{x,z+f(x)},
\]
where \(x\in\R^d\) and \(z\in\R\). The corresponding \emph{quantum global gradient oracle} is the unitary
\[
\cO_{\nabla f}\ket{x,w}
=\ket{x,w+\nabla f(x)},
\]
where \(x\in\R^d\) and \(w\in\R^d\).
\end{definition}

While widely adopted and analytically convenient, the global query model is too coarse to capture a basic computational feature of many statistical targets: their dependence structure is often local.
By locality, we mean that changing one coordinate, or a small block of coordinates, affects only a small part of the model---for example, only a few conditional interactions, neighboring variables, or local factors in the potential.
This form of locality has long been central in statistical modeling and computation.
It goes back to classical work on conditionally specified lattice systems and Gibbs distributions, where local conditional interactions serve as the basic modeling primitive \cite{Bes74,GG84}.
It later became foundational in Bayesian computation through Gibbs sampling and data augmentation methods \cite{GS90,TW87}, and subsequent work showed that scan order, blocking, and parameterization are closely tied to the underlying local dependence structure of the target distribution \cite{LWK95,RS97}.
A canonical modern example is provided by Gaussian Markov random fields: the target is Gaussian, hence log-concave, and its sparse precision matrix encodes which variables interact conditionally.
In this setting, sparse-matrix methods enable fast sampling, conditional simulation, and efficient block updates in Markov chain Monte Carlo (MCMC) \cite{Rue01}.
The same idea underlies the broader latent Gaussian literature, where sparse conditional dependence and sparse precision structure form the computational backbone of applications in spatial statistics, disease mapping, and spatio-temporal modeling \cite{RH05,RMC09,LRL11}.
A parallel theme appears in large-scale optimization, where this modeling locality has a direct computational analogue.
Nesterov motivates coordinate descent by observing that, in high-dimensional regimes, even simple full-dimensional vector operations can be prohibitively expensive, so one should exploit partial updates whenever they are substantially cheaper \cite{Nes12}.
This perspective has been developed and formalized in a large literature on coordinate and block-coordinate methods for smooth and composite objectives \cite{Tse01,TY09,BT13,RT14,Wri15}.
There, the basic computational unit is not a full gradient evaluation, but a coordinate gradient, a blockwise objective change, or a local update.
The resulting complexity guarantees are meaningful precisely because these local quantities can often be computed at much lower cost than global operations.

In summary, in both statistical sampling and optimization, local computational structure is not merely an implementation detail, but a central algorithmic resource.
By contrast, the standard quantum global query model completely obscures this structure.
If the target is accessed exclusively through global queries, the inherent cost distinction between a full-dimensional update and a local update is rendered invisible by design.
This structural disconnect naturally motivates the central question of this paper:

\begin{center}
\textit{Can quantum algorithms accelerate log-concave sampling by exploiting local structure?}
\end{center}

We answer this question in the affirmative for a broad and practically important class of structured log-concave targets.
The key point is that, in many such models, the relevant computational primitive is not a full evaluation of \(f\) or \(\nabla f\), but rather the effect of modifying a single coordinate while keeping the others fixed.
To capture this feature, we assume that the potential admits a local decomposition
\begin{equation} \label{eq:def-f}
f(x)=\sum_{a=1}^R \psi_a(x_{S_a}),\qquad S_a\subseteq[d],
\end{equation}
where each coordinate appears in only a small number of clauses.

In this paper, we mainly focus on those potentials with a \textit{constant bounded occurrence} local decomposition, where each coordinate appears in a constant number of clauses.
Far from being a narrow special case, this decomposition captures a common structural pattern in log-concave sampling problems arising from large-scale statistical computing and machine learning: many models of practical interest are built from sparse or local interactions, so changing a single coordinate affects only a small number of terms in the objective.
This situation arises naturally in Gaussian Markov random fields, conditional autoregressive models, and finite-range Gaussian lattice models, where sparse precision structure encodes local conditional dependence and variables interact only through local neighborhoods \cite{Bes74,Rue01,RH05}.
It also appears in latent Gaussian models for spatial statistics and spatio-temporal modeling, particularly in finite-element constructions based on stochastic partial differential equations (SPDEs) with locally supported basis functions and in spatial generalized linear mixed models with local observation structure \cite{RMC09,EMR09,LRL11}.
Banded-precision Gaussian models, such as state-space and Gauss--Markov models of bounded Markov order, provide another common source of log-concave targets with local quadratic potentials \cite{DAB19}.
Finally, generalized linear model posteriors and regularized empirical-risk objectives with sparse local design fit the same pattern in bounded-degree regimes, where each loss depends on few features and each feature enters few losses; this is the local-computation regime exploited by sparse coordinate methods \cite{FHT10,LLX15,ZX17}.
These examples motivate an access model in which local conditional computation is exposed directly, rather than mediated solely through global evaluations.

We now introduce, to our knowledge, the \emph{first} quantum local query model for continuous sampling that is designed to capture local structure. Rather than taking global evaluation or gradient queries as primitive, this model provides coherent access to the individual clauses in a local decomposition of \(f\).

\begin{definition}[Quantum local query]
\label{def:factor-oracle}
Suppose that the potential admits a local decomposition
$f(x)=\sum_{a=1}^R \psi_a(x_{S_a})$.
Without loss of generality, we assume that $S_a \in \binom{[d]}{k}$ for all $S_a$, i.e., each $S_a$ involves at most $k$ coordinates.
The \emph{quantum clause oracle} is the unitary
\[
\cO^{\mathrm{loc}}_f\ket{a,t,z}
=
\ket{a,t,\,z+\psi_a(t)},
\]
where \(a\in [R]\), \(t\in\R^k\), and \(z\in\R\). The corresponding \emph{clause-gradient oracle} is the unitary
\[
\cO^{\mathrm{loc}}_{\nabla f}\ket{a,t,w}
=
\ket{a,t,\,w+\nabla\psi_a(t)},
\]
where \(a\in [R]\), \(t\in\R^k\), and \(w\in\R^{ k}\).
\end{definition}

A local query acts on a \(k\)-dimensional clause
\(\psi_a:\R^k\to\R\), or its \(k\)-dimensional gradient, rather than on the full \(d\)-dimensional potential or gradient.
Thus, when \(k\) is small, the quantum clause and clause-gradient oracles can be implemented much more efficiently than their global counterparts.
Moreover, across many practical models, the local arity \(k\) is much smaller than \(d\), and often constant;
concrete instances are discussed in \cref{para:applications}.
This access model makes the local computational primitive of classical local-update methods visible at the level of quantum query complexity: a global evaluation must aggregate all clauses, whereas a single-coordinate conditional update only depends on the clauses involving that coordinate.
For the locally structured targets considered here, this makes local updates substantially cheaper than global evaluations or full gradient computations.

In this paper, we show that local queries provide the suitable computational interface for exploiting locality in quantum sampling.
They allow local structure to be converted into a genuine quantum algorithmic advantage, leading to an efficient quantum sampler for strongly log-concave distributions with local decompositions.

\subsection{Our results}
Our main result is an efficient quantum sampler in the local query model.
\begin{theorem}
    For any $L$-smooth and $\lambda$-strongly convex potential $f \colon \mathbb{R}^d \to \mathbb{R}$ with a constant bounded occurrence local decomposition, there is a quantum algorithm that produces a sample from a distribution that is $\varepsilon$-close in total variation (TV) distance to $\pi(x) \,\propto\, e^{-f(x)}$ using $\widetilde{O}(\sqrt{\kappa}d)$\footnote{We use $\widetilde O\parens{f}$ to represent
    $O \paren{f \cdot \polylog\parens{d,\kappa,1/\varepsilon}}$ throughout this paper
    to suppress polylogarithmic factors.} local queries, or $\widetilde{O}(\sqrt{\kappa d})$ queries when the initial distribution is warm.\footnote{
Here ``warm'' means that an initial quantum sample state for a distribution \(\mu_0\) is available, where \(\mu_0\) is constant-warm with respect to \(\pi\), i.e., \(\mathrm d\mu_0/\mathrm d\pi=O(1)\). We refer to this as the warm-start setting.
}
\end{theorem}
Table~\ref{tab:comparison-local-query} compares this result with representative classical and quantum samplers in the local query model for potentials with constant bounded occurrence local decompositions.
\begin{table}[htp]
\caption{
Comparison for strongly log-concave sampling in the local-decomposition setting.
Complexities are measured in local queries.
For algorithms based on global evaluations or full gradients, each such query costs \(\Theta(d)\) local queries.
We include representative recent approaches, including several quantum samplers; given the large classical MCMC literature, the comparison is not intended to be exhaustive.
A blank warm-start entry means that no polynomially improved warm-start complexity is included; \(W_2\) denotes the \(2\)-Wasserstein distance.
}
  \label{tab:comparison-local-query}
  \centering
  \small
  \resizebox{\textwidth}{!}{
\begin{tabular}{@{}llll@{}}
  \toprule
  Method & Guarantee\({}^{\star}\) & Feasible/annealed start & Warm start \\
  \midrule
    ULD-RMM \cite{SL19}
    & \(W_2\)
    & \(\widetilde O\paren{
        \tfrac{\kappa^{7/6}d^{7/6}}{\varepsilon^{1/3}}
        +
        \tfrac{\kappa d^{4/3}}{\varepsilon^{2/3}}
      }\)
    & -- \\
  RWM \cite{ALPW22}
    & TV
    & \(\widetilde O\paren{\kappa d^2}\)
    & --\\
  MALA \cite{LST20,WSC22}
    & TV
    & \(\widetilde O\paren{\kappa d^2}\)
    & \(\widetilde O\paren{\kappa d^{3/2}}\) \\
  HMC \cite{CDWY20,CGJ23}
    & TV
    & \(\widetilde O\paren{\kappa d^{23/12}}\)
    & \(\widetilde O\paren{\kappa d^{5/4}}\) \\
     Gibbs sampler \cite{ALZ24}
    & TV
    & \(\widetilde O\paren{\kappa d}\)
    & --\\
  \midrule
      Quantum inexact ULD \cite{CLL22}
    & \(W_2\)
    & \(\widetilde O\paren{
        \tfrac{\kappa^2 d^{3/2}}{\varepsilon}
      }\)
    & -- \\
  Quantum inexact ULD-RMM \cite{CLL22}
    & \(W_2\)
    & \(\widetilde O\paren{
        \tfrac{\kappa^{7/6}d^{7/6}}{\varepsilon^{1/3}}
        +
        \tfrac{\kappa d^{4/3}}{\varepsilon^{2/3}}
      }\)
    & -- \\
  Quantum finite-sum SVRG-HMC \cite{OLMW25}
    & \(W_2\)
    & \(\widetilde O\paren{
        \tfrac{Ld^{1/2}\kappa^{3/2}}{\varepsilon}
        +
        \tfrac{L^{9/8}d^{11/8}\kappa^{3/4}}{\varepsilon^{3/4}}
      }^{\dagger}\)
    & -- \\
      Quantum finite-sum CV-HMC \cite{OLMW25}
    & \(W_2\)
    & \(\widetilde O\paren{
        \tfrac{Ld^{5/4}\kappa^{9/4}}{\varepsilon^{3/2}}
      }^{\dagger}\)
    & -- \\
          Quantum MALA \cite{CLL22}
    & TV
    & \(\widetilde O\paren{\sqrt{\kappa}\,d^2}\)
    & \(\widetilde O\paren{\sqrt{\kappa}\,d^{5/4}}\) \\
      Quantum operator-level sampler \cite{LDCL26}
    & TV
    & --
    & \(\widetilde O\paren{\sqrt{\kappa}\,d^{3/2}}^{\ddagger}\) \\
  \textbf{This work}
    & TV
    & {\(\bm{\widetilde O\paren{\sqrt{\kappa}\,d}}\)}
    & {\(\bm{\widetilde O\paren{\sqrt{\kappa d}}}\)} \\
  \bottomrule
\end{tabular}
}
 \vspace{2pt}
  \begin{flushleft}
\footnotesize
\({}^{\star}\) In finite precision, TV and the \(W_2\) are related by grid-dependent factors: on a grid of diameter \(D\) and minimum spacing \(h\), \(h\sqrt{\TV}\le W_2\le D\sqrt{\TV}\).
For algorithms with only logarithmic accuracy dependence and polynomially bounded \(D,h^{-1}\), this conversion is absorbed by the \(\widetilde O\) notation.
We therefore keep the native TV or \(W_2\) guarantee of each cited work.

\({}^{\dagger}\) The bound of \cite{OLMW25} is in a finite-sum component gradient model; we count one component gradient query as one local query, so no additional \(O(d)\) conversion is applied.

\({}^{\ddagger}\) The native bound of \cite{LDCL26} is
\(\widetilde O\parens{d^{1/2}C_{\rm PI}^{1/2}}\), where \(C_{\rm PI}\) is the Poincar\'e constant.
By the Brascamp--Lieb inequality, \(C_{\rm PI}\le 1/\lambda\) for \(\lambda\)-strongly log-concave targets; after normalizing \(L=1\), this gives \(C_{\rm PI}^{1/2}\le \sqrt{\kappa}\).
\end{flushleft}
\end{table}
We also consider the more general setting where each coordinate may appear in up to \(\Delta\) clauses.
\begin{theorem}
    For any $L$-smooth and $\lambda$-strongly convex potential $f \colon \mathbb{R}^d \to \mathbb{R}$ with a local decomposition where each coordinate appears in at most $\Delta$ clauses, there is a quantum algorithm that produces a sample from a distribution that is $\varepsilon$-close in TV distance to $\pi(x) \propto\, e^{-f(x)}$ using $\widetilde{O}(\sqrt{\kappa}d\Delta)$ local queries, or $\widetilde{O}(\sqrt{\kappa d}\Delta)$ queries when the initial distribution is warm.
\end{theorem}

\paragraph{Applications.}\label{para:applications}
We spell out several representative model classes in which the local decomposition is explicit: each clause depends only on a small subset of coordinates, and each coordinate participates in only a constant number of clauses.

First, consider Gaussian Markov random fields.
Let \(G=([d],E)\), and let \(Q\succ 0\) be a sparse precision matrix with \(Q_{ij}=0\) whenever \(\{i,j\}\notin E\).
For \(\pi(x)\propto \exp\{-\frac12 x^\top Qx+h^\top x\}\) with \(\lambda I\preceq Q\preceq L I\), the potential admits the decomposition
$    f(x)
    =
    \sum_{i=1}^d
    \left(
        \frac12 Q_{ii}x_i^2 - h_i x_i
    \right)
    +
    \sum_{\{i,j\}\in E}
        Q_{ij}x_i x_j$.
Thus the clauses are unary or pairwise, and each coordinate appears only in the terms involving its neighbors.
For bounded-degree Gaussian Markov random fields, including finite-range lattice models and conditional autoregressive models in fixed spatial dimension, the occurrence number is constant \cite{Bes74,Rue01,RH05}.
Banded-precision Gaussian Markov models of bounded Markov order are a one-dimensional instance of the same principle: if \(Q_{ij}=0\) for \(|i-j|>b\) with constant bandwidth \(b\), then each coordinate interacts only with a bounded temporal neighborhood, matching the banded-operator structure studied in \cite{DAB19}.
In these settings, our theorem gives local query complexity \(\widetilde O(\sqrt{\kappa}\,d)\), or \(\widetilde O(\sqrt{\kappa d})\) from a warm start.

Second, SPDE/finite-element latent Gaussian models also lead to local decompositions.
In the SPDE approach, a latent Gaussian field is represented by coefficients \(x\in\R^d\) in a locally supported finite-element basis, and sparsity of the precision matrix comes from the fact that basis functions overlap only locally \cite{LRL11}.
Equivalently, the prior energy can be written as
$    f_{\rm prior}(x)
    =
    \sum_{e\in\mathcal T}
        \psi_e(x_{V(e)})$,
where \(e\) ranges over mesh elements and \(V(e)\) denotes the vertices of element \(e\).
For fixed-order elements in fixed spatial dimension, each element involves only a constant-size set of coefficients.
Thus, on bounded-degree meshes, the finite-element prior has a constant-occurrence local decomposition.
When the target remains smooth and strongly log-concave, our theorem applies.

Another important class is given by sparse generalized linear models and regularized empirical-risk objectives.
A typical log-concave posterior or Gibbs density has potential
$   f(\theta)
    =
    \sum_{i=1}^n
        \ell_i(a_i^\top \theta; y_i)
    +
    \sum_{j=1}^d
        r_j(\theta_j)$,
where \(\ell_i\) is a convex smooth loss and the regularizer or prior supplies strong convexity.
If the design vector \(a_i\) has support \(S_i\), then the loss term \(\ell_i(a_i^\top\theta;y_i)\) depends only on \(\theta_{S_i}\).
Thus bounded row sparsity gives low-dimensional local clauses, while bounded column sparsity ensures that each coordinate participates in only a bounded number of losses.
This is precisely the sparse-coordinate regime exploited by coordinate methods for generalized linear models and regularized empirical-risk minimization \cite{FHT10,LLX15,ZX17}.
In the bounded row- and column-sparsity setting, our theorem yields local query complexity \(\widetilde O(\sqrt{\kappa}\,d)\), or \(\widetilde O(\sqrt{\kappa d})\) from a warm start.

\subsection{Techniques}
Our algorithm follows a different route from prior quantum samplers.

Existing quantum approaches to continuous sampling largely build on global Langevin-, Hamiltonian Monte Carlo (HMC)-, or Metropolis-type dynamics.
In \cite{CLL22}, the quantum underdamped Langevin dynamics (ULD) and randomized midpoint method (RMM) algorithms use quantum zeroth-order information to implement inexact first-order Langevin updates, while the quantum Metropolis-adjusted Langevin algorithm (MALA) assumes full gradient access and applies a quantum-walk speedup to the Metropolis-adjusted Langevin chain.
Subsequent work follows a related perspective in more stochastic settings: \cite{OLMW24} uses stochastic gradients to implement quantum walks for finite-sum non-logconcave targets, and \cite{OLMW25} develops quantum variance-reduction and stochastic-gradient techniques for HMC/Langevin Monte Carlo (LMC)-type samplers.
The operator-level approach of \cite{LDCL26} is different from these methods: it encodes continuous Langevin-type operators through the Witten Laplacian and uses singular-value transformation, but still treats the target through global potential or gradient information.
In contrast, our approach starts from the local conditional-update structure exposed by the local query model, mirroring the locality long exploited by classical Gibbs sampling and coordinate-based methods.

The local query model suggests a different quantum primitive: instead of building on a global diffusion or Metropolis proposal, we coherently implement the one-coordinate Gibbs update.
For a fixed coordinate \(m\), freezing \(x_{-m}\) reduces the update to the one-dimensional conditional density
\[
\pi_m(u\mid x_{-m})\,\propto\,
\exp\!\left(-f_m(u;x_{-m})\right),
\qquad
f_m(u;x_{-m})
=
\sum_{a:m\in S_a}\psi_a\bigl((u,x_{-m})_{S_a}\bigr).
\]
Only the clauses involving coordinate \(m\) enter this conditional potential, so the update is intrinsically local.
The first step is to implement this one-dimensional conditional update coherently.
We locate the conditional mode by evaluating the one-dimensional derivative \(\partial_u f_m\) from the incident clause-gradient oracles, recenter and rescale the conditional potential so that it becomes \(1\)-strongly convex and \(\kappa\)-smooth, and then apply the one-dimensional rejection-envelope sampler of \cite{CGL22}.
Since the envelope construction is deterministic, it can be performed reversibly; the remaining accept-reject step is made coherent using fixed-point amplitude amplification \cite{YLC14}.
This gives a coherent conditional sampler whose query cost is proportional to the number of clauses incident to the updated coordinate.

Once this one-dimensional coherent conditional sampler is available, we build the corresponding continuous-space Szegedy walk \cite{Sze04,CCH23} for the classical random-scan Gibbs kernel.
The main classical ingredient is the recent entropy-contraction theory of \cite{ALZ24}, which proves a sharp contraction of relative entropy for the random-scan Gibbs sampler under log-concavity.
In the \(L\)-smooth and \(\lambda\)-strongly log-concave setting, this entropy-contraction result implies the spectral-gap bound
$\gap(P^{\mathrm{GS}})
=\Omega({1}/{\kappa d})$.
For the corresponding Szegedy walk, the eigenphases are related to the singular values of the classical discriminant operator; in particular, this gap yields a quantum phase gap
$\Omega({1}/{\sqrt{\kappa d}})$.
We then combine these Gibbs walks with a Gaussian cooling schedule.
The schedule has \(\widetilde O(\sqrt d)\) stages with constant overlap between consecutive quantum sample states.
Using the slowly-varying quantum walk framework of \cite{SBBK08,WA08}, the final quantum sample state is prepared using $\widetilde O(\sqrt d\cdot\sqrt{\kappa d})
=\widetilde O(\sqrt{\kappa}\,d)$
walk applications.
Each walk application costs \(\widetilde O(\Delta)\) local queries, where \(\Delta\) is the maximum coordinate occurrence number.
This gives the query complexity
$\widetilde O(\Delta\sqrt{\kappa}d)$.
With a warm initial distribution, the cooling schedule is unnecessary, and the cost becomes
$\widetilde O(\Delta\sqrt{\kappa d})$.
\subsection{Discussion}
This work shows that local structure can be a genuine algorithmic resource for quantum sampling.
By replacing global potential or gradient access with coherent access to local clauses, we obtain a quantum sampler whose cost scales with the number of clauses affected by a coordinate update, rather than with the cost of evaluating the full potential.
For bounded-occurrence local decompositions, this leads to local query complexity \(\widetilde O(\sqrt{\kappa}\,d)\), and \(\widetilde O(\sqrt{\kappa d})\) from a warm start.

More broadly, the local query model offers a new perspective on the design of quantum sampling algorithms.
Instead of treating the target distribution as a black-box global potential, it encourages algorithms that interact directly with the conditional and combinatorial structure of the model.
We highlight several particularly interesting directions suggested by this viewpoint.

A first direction is to understand optimality in the local query model.
It would be valuable to prove lower bounds clarifying whether the dependence on \(d\), \(\kappa\), and the occurrence number \(\Delta\) is intrinsic, or whether further quantum speedups are possible under stronger local structure.
Relatedly, our bounds are stated in terms of the worst-case occurrence number \(\Delta\).
For highly nonuniform factor graphs, it would be interesting to develop degree-adaptive algorithms whose complexity depends on average or weighted local degrees, possibly using nonuniform scan rules or variable-time implementations.

Another natural direction is to enrich the local-update primitive beyond one-coordinate random-scan Gibbs updates.
Block Gibbs updates, graph-coloring schemes, and nonuniform coordinate choices can change the balance between the mixing behavior of the underlying chain and the cost of implementing each conditional sampler.
Such variants are well studied in classical Gibbs sampling and local-update methods, where scan order, blocking, and graph-coloring constructions can substantially affect computational behavior \cite{LWK95,RS97,GLGG11,HDSMR16}.
Understanding the analogous tradeoffs in the quantum setting could lead to sharper algorithms for structured graphical models.
Similarly, one may hope to exploit coordinate-wise smoothness, local curvature, arbitrary sampling rules, or blockwise condition numbers instead of relying only on the global condition number \(\kappa\), paralleling classical coordinate and block-coordinate methods \cite{Nes12,RT14,Wri15,QR16}.

Finally, the local query perspective should not be limited to smooth strongly log-concave targets.
An important challenge is to extend the framework to distributions governed by weaker functional inequalities, such as Poincar\'e, log-Sobolev, or related interpolation conditions.
This would require both coherent conditional samplers with suitable accuracy guarantees and quantum-walk phase-gap bounds for the corresponding local-update chains; such functional-inequality assumptions have become central in recent analyses of Langevin-type sampling beyond the strongly log-concave setting \cite{CEL25}.
It would also be interesting to develop local query algorithms for nonsmooth composite potentials, using ideas such as smoothing, proximal updates, or regularization, in line with recent work on nonsmooth and composite sampling \cite{MFWB22}.

\section{Preliminaries}
We consider the standard one-coordinate random-scan Gibbs sampler for the target distribution
\(
\pi(x)=Z_f^{-1}e^{-f(x)}
\)
on \(\R^d\), where \(f:\R^d\to\R\) is \(L\)-smooth and \(\lambda\)-strongly convex. For each coordinate \(m\in[d]\), let \(x_{-m}\) denote all coordinates except \(x_m\), and let \(\pi_m(\cdot\mid x_{-m})\) be the conditional distribution of the \(m\)-th coordinate under \(\pi\). Equivalently, \(\pi_m(\cdot\mid x_{-m})\) is the one-dimensional distribution on \(\R\) with density proportional to
\(
u\mapsto \exp \bigl(-f(x_1,\dots,x_{m-1},u,x_{m+1},\dots,x_d)\bigr).
\)
The coordinate-\(m\) Gibbs update is the Markov kernel \(P_m\) that leaves all coordinates except \(m\) unchanged and resamples \(x_m\) exactly from \(\pi_m(\cdot\mid x_{-m})\). The random-scan Gibbs kernel is then
$P^{\mathrm{GS}}:=\frac{1}{d}\sum_{m=1}^d P_m.$
Equivalently, one step of the chain first draws \(m\sim \mathrm{Unif}([d])\), and then replaces \(x_m\) by an exact sample from \(\pi_m(\cdot\mid x_{-m})\).
Since \(f\) is \(L\)-smooth and \(\lambda\)-strongly convex, each one-dimensional conditional distribution is itself \(L\)-smooth and \(\lambda\)-strongly log-concave.

\section{Quantum random-scan Gibbs sampler}
\label{sec:q-random-scan-gibbs}
In this section we develop a quantum random-scan Gibbs sampler in the local query model.
Rather than starting from a global Langevin-, HMC-, or Metropolis-type dynamics, our construction builds directly on the coordinate-conditional structure exposed by the local decomposition.
Indeed, when a coordinate \(m\in[d]\) is resampled, only the clauses incident to \(m\) can affect the conditional law.
We denote this incident set by \(D_m:=\{a\in[R]:m\in S_a\}\), and write \(\Delta:=\max_{m\in[d]}|D_m|\) for the occurrence number.

The section proceeds in two steps.
In \cref{subsec:coherent-conditional-sampler}, we construct a coherent one-dimensional conditional sampler for \(\pi(\cdot\mid x_{-m})\), for a fixed coordinate \(m\) and current state \(x\).
This is the basic local primitive.
In \cref{subsec:quantum-gibbs-walk}, we use this primitive to implement the Szegedy-type quantum walk associated with the classical random-scan Gibbs kernel.

\subsection{Coherent one-coordinate conditional sampler}
\label{subsec:coherent-conditional-sampler}

Fix a coordinate \(m\in[d]\) and a state \(x\in\R^d\).
Freezing all coordinates except \(x_m\), the conditional potential is
\[
    f_m(u;x_{-m})
    :=
    \sum_{a\in D_m}
        \psi_a\bigl((u,x_{-m})_{S_a}\bigr),
    \qquad u\in\R .
\]
The corresponding conditional density is \(\pi_m(u\mid x_{-m})=Z_m(x_{-m})^{-1}\exp\{-f_m(u;x_{-m})\}\).
Since \(f\) is \(L\)-smooth and \(\lambda\)-strongly convex on \(\R^d\), the one-dimensional potential \(u\mapsto f_m(u;x_{-m})\) is also \(L\)-smooth and \(\lambda\)-strongly convex.
It therefore has a unique minimizer, denoted by \(u_m^\star(x_{-m}):=\arg\min_{u\in\R} f_m(u;x_{-m})\).
The key point is locality: both \(f_m\) and \(\partial_u f_m\) only involve the incident clauses \(D_m\), and hence can be evaluated using \(O(|D_m|)\) local queries; see \Cref{prop:local-conditional-access}.

\paragraph{Normalizing one-dimensional conditional targets.}
Although each Gibbs update is one-dimensional, the conditional target is not yet in a uniform algorithmic form: as the ambient state \(x\) varies, the conditional potential \(f_m(\cdot;x_{-m})\) may shift in location and scale.
We therefore normalize the conditional target before sampling from it.
The natural reference point is the conditional minimizer \(u_m^\star(x_{-m})\), equivalently the mode of \(\pi_m(\cdot\mid x_{-m})\).
Centering near this point removes the state-dependent translation, while rescaling by the strong-convexity parameter reduces every conditional law to a common one-dimensional class whose geometry is controlled only by the condition number \(\kappa\).
This is analogous to local curvature normalization in coordinate methods \cite{Nes12,RT14}, and matches the one-dimensional rejection-envelope framework of \cite{CGL22}.

We begin with a coherent mode-location primitive.
Its role is not to sample from the conditional law, but to compute a local center around which the conditional target can be normalized.
For a fixed coordinate \(m\in[d]\), let \(U_m^{\mathrm{mode}}\) denote the unitary
$U_m^{\mathrm{mode}}\ket{x}\ket{0}_{\mathrm c}\ket{0}_{\mathrm w}
=
\ket{x}\ket{\widehat u_m(x)}_{\mathrm c}\ket{0}_{\mathrm w}$,
where \(\widehat u_m(x)\) approximates the conditional minimizer \(u_m^\star(x_{-m})\), the register \(\mathrm c\) stores the output center, and \(\mathrm w\) is workspace.
The mode-location subroutine is described in \Cref{alg:q-mode-locate}.
The guarantee needed for the coherent conditional sampler is stated below, with the proof deferred to \Cref{app:proof-mode-location}.

\begin{proposition}[Quantum mode-location primitive]
\label{prop:q-mode-location}
For every fixed \(m\in[d]\), Algorithm~\ref{alg:q-mode-locate} implements a unitary \(U_m^{\mathrm{mode}}\) satisfying
\(U_m^{\mathrm{mode}}\ket{x}\ket{0}_{\mathrm c}\ket{0}_{\mathrm w}
=
\ket{x}\ket{\widehat u_m(x)}_{\mathrm c}\ket{0}_{\mathrm w}\),
where \(\widehat u_m(x)\) is an \(L^{-1/2}\)-accurate approximation to the conditional minimizer, namely
\(|\widehat u_m(x)-u_m^\star(x_{-m})|\le L^{-1/2}\).
Its implementation uses
\(O\!\left(
|D_m|
\log\!\bigl(1+\frac{\sqrt L}{\lambda}|\partial_m f(x)|\bigr)
\right)\)
queries to the local clause-gradient oracle.
Furthermore, conditioned on the output \(\widehat u_m(x)\), define the normalized conditional potential
\[
\bar f_m(v;x)
:=
f_m\!\left(\widehat u_m(x)+\frac{v}{\sqrt{\lambda}};x_{-m}\right)
-
f_m\!\left(\widehat u_m(x);x_{-m}\right).
\]
Then \(\bar f_m(\cdot;x)\) is \(1\)-strongly convex and \(\kappa\)-smooth, and its minimizer \(v_m^\star(x)\) satisfies
\(|v_m^\star(x)|\le \kappa^{-1/2}\).
Moreover, one evaluation of \(\bar f_m(\cdot;x)\), respectively \(\partial_v\bar f_m(\cdot;x)\), can be implemented using \(O(|D_m|)\) queries to the local clause oracle, respectively the local clause-gradient oracle.
\end{proposition}

\paragraph{Quantum one-dimensional conditional sampler.}

We now turn the normalized conditional target from \Cref{prop:q-mode-location} into a coherent sampler.
After the mode-location step, the conditional target has the canonical form \(e^{-\bar f_m(v;x)}\): the potential \(\bar f_m(\cdot;x)\) is \(1\)-strongly convex and \(\kappa\)-smooth, and its minimizer lies within distance \(\kappa^{-1/2}\) of the origin.
This is precisely the one-dimensional regime in which the rejection-envelope sampler of \cite{CGL22} applies.

The point of the construction is to make this one-dimensional sampler compatible with the local query model.
The one-dimensional sampler of \cite{CGL22} takes oracle access to a potential and deterministically constructs a description \(\Theta\) of an envelope \(\widetilde q_\Theta\) and the associated proposal distribution.
Thus, once the normalized conditional oracle is available coherently, this construction can be carried out reversibly.
Moreover, by \Cref{prop:q-mode-location}, each query to the normalized target \(\bar f_m(\cdot;x)\) is local: it only touches the clauses incident to coordinate \(m\).
Thus the resulting coherent conditional sampler has query cost proportional to \(|D_m|\), rather than to the total number of clauses. We use the following form of the one-dimensional construction.

\begin{lemma}[One-dimensional envelope construction~\cite{CGL22}]
\label{lem:1d-envelope}
Let \(\phi:\R\to\R\) be \(1\)-strongly convex and \(\kappa\)-smooth, and suppose that its minimizer lies within distance \(\kappa^{-1/2}\) of the origin.
Given one-dimensional oracle access to \(\phi\), there is a deterministic procedure using \(O(\log\log\kappa)\) oracle queries that outputs finite data \(\Theta\) describing an explicit envelope \(\widetilde q_\Theta\) for \(\widetilde p(v):=e^{-\phi(v)}\).
The envelope satisfies
\[
\widetilde q_\Theta(v)\ge \widetilde p(v)
\qquad\text{and}\qquad
\int \widetilde q_\Theta(v)\,\d v
\le C
\int \widetilde p(v)\,\d v
\]
for a universal constant \(C>0\).
Moreover, after normalizing \(q_\Theta:=\widetilde q_\Theta/\int \widetilde q_\Theta\), the proposal \(q_\Theta\) can be sampled explicitly, and the acceptance function
\(\alpha_\Theta(v):={\widetilde p(v)}/{\widetilde q_\Theta(v)}\)
can be evaluated from \(\Theta\) and oracle access to \(\phi\).
\end{lemma}

Applying \Cref{lem:1d-envelope} to \(\phi=\bar f_m(\cdot;x)\) gives envelope data \(\Theta_{m,x}\), a proposal density \(q_{m,x}\), and an acceptance function \(\alpha_{m,x}\).
Since the envelope normalizing constant is within a universal constant factor of the target normalizing constant, the corresponding rejection sampler has constant success probability.
The deterministic data \(\Theta_{m,x}\) can be computed reversibly, and the accept-reject step can be made coherent.
The remaining rejection step is implemented by fixed-point amplification of the accepting branch \cite{YLC14}.

\begin{theorem}[Coherent one-dimensional conditional sampler]
\label{thm:qcondsample}
For every fixed \(m\in[d]\) and \(\eta\in(0,1)\), \Cref{alg:qcondsample-clean} implements a unitary \(U_{m,\eta}^{\mathrm{cond}}\) such that, for every \(x\in\R^d\),
\[
\left\|
U_{m,\eta}^{\mathrm{cond}}
\ket{x}\ket{0}_{\mathrm c}\ket{0}_{\mathrm{out}}\ket{0}_{\mathrm w}
-
\ket{x}\ket{0}_{\mathrm c}
\int \sqrt{\pi_m(u\mid x_{-m})}\ket{u}_{\mathrm{out}}\,\d u
\ket{0}_{\mathrm w}
\right\|_2
\le
\eta .
\]
On input \(x\), the implementation uses
\(\widetilde O(|D_m|)\)
local queries.
\end{theorem}
The corresponding procedure is given in \Cref{alg:qcondsample-clean}, and the proof is deferred to \Cref{app:proof-conditional-sampler}.

\subsection{Quantum walk for the Gibbs kernel}
\label{subsec:quantum-gibbs-walk}

We now lift the coherent conditional sampler to the full random-scan Gibbs dynamics.
Putting the coordinate index in a uniform superposition and applying the conditional sampler coherently gives the update isometry for the random-scan Gibbs kernel.
The associated Szegedy walk has a phase gap governed by the classical spectral gap of this kernel.
Finally, we combine these Gibbs walks along a Gaussian cooling schedule to prepare the quantum sample state \(\ket{\pi}:=\int\sqrt{\pi(x)}\ket{x}\,\d x\).

\paragraph{Coherent Gibbs update isometry.}
We first assemble the coordinate conditional samplers from \Cref{thm:qcondsample} into the isometry corresponding to one random-scan Gibbs update.
For \(m\in[d]\), let \(U_{m,\eta}^{\mathrm{cond}}\) be the coherent conditional sampler for coordinate \(m\), implemented to accuracy \(\eta\), and define the controlled sampler
\(U_{\eta}^{\mathrm{cond}}:=\sum_{m=1}^d \ket m\!\bra m\otimes U_{m,\eta}^{\mathrm{cond}}\).
For \(u\in\R\), write \(x^{(m\leftarrow u)}\) for the vector obtained from \(x\) by replacing \(x_m\) with \(u\).
The ideal coherent state for one random-scan Gibbs update from \(x\) is
\[
\ket{\psi_x^{\mathrm{GS}}}
:=
\frac{1}{\sqrt d}
\sum_{m=1}^d
\ket m
\int \sqrt{\pi_m(u\mid x_{-m})}\,
\ket{x^{(m\leftarrow u)}}\,\d u .
\]
Equivalently, the ideal update isometry is
\(T_\star^{\mathrm{GS}}\ket{x}:=\ket{x}\ket{\psi_x^{\mathrm{GS}}}\), where the first register stores the current state and the second register stores the scan label together with the updated state.
The following lemma, proved in \Cref{app:proof-gibbs-isometry}, formalizes the implemented isometry and its local query cost.

\begin{lemma}[Coherent Gibbs update isometry]
\label{lem:qgibbs-update}
Assume that the coordinate conditional samplers \(U_{m,\eta}^{\mathrm{cond}}\) are implemented with operator-norm error at most \(\eta\), uniformly over \(m\in[d]\).
Then there is an implemented isometry \(\widetilde T_{\eta}^{\mathrm{GS}}\) satisfying
$\|
\widetilde T_{\eta}^{\mathrm{GS}}
-
T_\star^{\mathrm{GS}}
\|_{\mathrm{op}}
\le \eta $.
Moreover, if the coordinate-\(m\) conditional sampler has local query cost \(\widetilde O(|D_m|)\), then one application of \(\widetilde T_{\eta}^{\mathrm{GS}}\) or \(\widetilde T_{\eta}^{\mathrm{GS}}{}^\dagger\) uses \(\widetilde O(\Delta)\) local queries.
\end{lemma}
The phase gap analysis relies on a recent breakthrough in the theory of Gibbs sampling: the entropy-contraction result of \cite{ALZ24}.
For a reversible Markov kernel \(P\) with stationary distribution \(\pi\), write \(\gap(P)\) for its spectral gap.
We use the following consequence for the Gibbs sampler.
\begin{lemma}[Spectral gap of random-scan Gibbs {\cite[Corollary~3.7]{ALZ24}}]
\label{lem:gibbs-classical-gap}
For the one-coordinate random-scan Gibbs kernel \(P^{\mathrm{GS}}\) targeting an \(L\)-smooth and \(\lambda\)-strongly log-concave density on \(\R^d\), one has
$\gap(P^{\mathrm{GS}})
=
\Omega\!\left(\frac{1}{\kappa d}\right)$.
\end{lemma}

Let \(\Pi_\star^{\mathrm{GS}}:=T_\star^{\mathrm{GS}}(T_\star^{\mathrm{GS}})^\dagger\).
Define the swap \(S_{\mathrm{GS}}\ket{x}\ket m\ket y=\ket y\ket m\ket x\), and define the ideal Gibbs walk by \(W_\star^{\mathrm{GS}}:=S_{\mathrm{GS}}(2\Pi_\star^{\mathrm{GS}}-I)\).
The implemented walk \(\widetilde W_{\eta}^{\mathrm{GS}}\) is obtained by replacing \(T_\star^{\mathrm{GS}}\) with \(\widetilde T_{\eta}^{\mathrm{GS}}\) in the reflection.
Let \(\gamma_{\mathrm{GS}}\) be the smallest nonzero absolute eigenphase of \(W_\star^{\mathrm{GS}}\).
The following lemma, proved in \Cref{app:proof-gibbs-spectrum}, identifies the discriminant of this quantum walk.

\begin{lemma}[Phase gap of the Gibbs quantum walk]
\label{lem:qgibbs-spectrum}
The discriminant operator induced by \(T_\star^{\mathrm{GS}}\) and \(S_{\mathrm{GS}}\) is the Szegedy discriminant of the classical random-scan Gibbs kernel \(P^{\mathrm{GS}}\).
Consequently, the nonzero eigenphases of \(W_\star^{\mathrm{GS}}\) are separated from \(0\) by \(\Omega(\sqrt{\gap(P^{\mathrm{GS}})})\).
In particular, \(\gamma_{\mathrm{GS}}=\Omega(1/\sqrt{\kappa d})\).
Moreover, if \(\|\widetilde T_{\eta}^{\mathrm{GS}}-T_\star^{\mathrm{GS}}\|_{\mathrm{op}}\le \eta\), then \(\|\widetilde W_{\eta}^{\mathrm{GS}}-W_\star^{\mathrm{GS}}\|_{\mathrm{op}}=O(\eta)\).
\end{lemma}

\paragraph{Annealed state preparation.}
The Gibbs walk above gives an efficient way to move between nearby quantum sample states.
To obtain a sampler from a simple initial state, we use a Gaussian cooling schedule, a standard annealing tool in log-concave sampling and volume computation \cite{LV07,CV18}.
Starting from an easily prepared Gaussian, the schedule gradually removes a Gaussian regularizer while keeping consecutive distributions well overlapped.
The slowly varying quantum-walk framework of \cite{WA08} then converts the corresponding sequence of Gibbs walks into a quantum state-preparation procedure.

Let \(x^\star\) be the minimizer of \(f\).
For notational simplicity, translate coordinates so that \(x^\star=0\), and subtract the constant \(f(x^\star)\).
Let \(\sigma_1^2=1/(2Ld)\) and \(\sigma_{i+1}^2=(1+1/\sqrt d)\sigma_i^2\).
Choose \(s\) to be the first index for which \(\sigma_s^2=C d/\lambda\), where \(C>0\) is a sufficiently large universal constant.
Define \(\pi_0(x)\propto \exp(-\|x\|^2/(2\sigma_1^2))\), define \(\pi_i(x)\propto \exp(-f(x)-\|x\|^2/(2\sigma_i^2))\) for \(1\le i\le s\), and set \(\pi_{s+1}:=\pi\).
Then \(s=\widetilde O(\sqrt d)\).
Since \(\|x\|^2/(2\sigma_i^2)=\sum_{j=1}^d x_j^2/(2\sigma_i^2)\), the added Gaussian regularizer contributes only one unary clause per coordinate.
Thus the occurrence number increases by at most one along the cooling schedule.
For each stage \(i\), let \(P_i^{\mathrm{GS}}\), \(T_{\star,i}^{\mathrm{GS}}\), and \(W_{\star,i}^{\mathrm{GS}}\) denote the random-scan Gibbs kernel, the ideal update isometry, and the ideal Gibbs walk associated with \(\pi_i\).
We also write \(\ket{\pi_i}:=\int\sqrt{\pi_i(x)}\ket{x}\,\d x\).
The following lemma, proved in \Cref{app:proof-cooling}, records the two properties of the schedule needed for quantum state conversion.

\begin{lemma}[Gaussian cooling schedule for Gibbs walks]
\label{lem:qgibbs-overlap}
The above schedule has \(s+1=\widetilde O(\sqrt d)\) transitions, and adjacent quantum sample states satisfy \(|\langle \pi_i\mid \pi_{i+1}\rangle|=\Omega(1)\) for all \(i=0,\dots,s\).
Moreover, for every stage \(i\), the random-scan Gibbs kernel satisfies \(\gap(P_i^{\mathrm{GS}})=\Omega(1/(\kappa d))\), and the corresponding Gibbs walk has phase gap \(\gamma_{\mathrm{GS},i}=\Omega(1/\sqrt{\kappa d})\).
\end{lemma}

The next theorem gives the resulting sampler.
Its proof is deferred to \Cref{app:proof-main-sampler}; the explicit procedure is given in \Cref{alg:qrsgs}.
The corresponding warm-start result, which avoids the cooling schedule, is stated and proved in \Cref{app:proof-warm-start}.
\begin{theorem}[Quantum random-scan Gibbs sampler]
\label{thm:q-rsgs-main}
Let \(f:\R^d\to\R\) be \(L\)-smooth and \(\lambda\)-strongly convex, and suppose that \(f\) admits a local decomposition with occurrence number \(\Delta\).
In the quantum local query model, \(\Cref{alg:qrsgs}\) prepares a quantum sample state \(\ket{\widetilde\pi}\) satisfying
\[
    \|\ket{\widetilde\pi}-\ket{\pi}\|_2\le \varepsilon,
    \qquad
    \ket{\pi}:=\int\sqrt{\pi(x)}\ket{x}\,\d x .
\]
Its local query complexity is \(\widetilde O(\Delta\sqrt{\kappa}\,d)\).
Consequently, measuring \(\ket{\widetilde\pi}\) produces a classical sample from a distribution \(q\) satisfying \(\TV(q,\pi)\le \varepsilon\).
\end{theorem}

\bibliographystyle{alpha}
\bibliography{ref}

\appendix
\crefalias{section}{appendix}
\crefalias{subsection}{subappendix}
\crefalias{subsubsection}{subsubappendix}

\section{Omitted Algorithm Details}
\label{app:pseudocode}

This appendix collects the pseudocode for the three algorithmic primitives used in the main text.
The first procedure computes an approximate conditional mode, which is used as the center for the normalized one-dimensional conditional target.

\begin{algorithm}[H]
\caption{\(\textsc{QModeLocate}_m\)}
\label{alg:q-mode-locate}
\begin{algorithmic}[1]
\Statex \textbf{Input:} \(\ket{x}\ket{0}_{\mathrm c}\ket{0}_{\mathrm w}\).
\Statex \textbf{Output:} \(\ket{x}\ket{\widehat u_m(x)}_{\mathrm c}\ket{0}_{\mathrm w}\).
\State Coherently compute \(g_m\gets \partial_m f(x)\).
\State Set \(r\gets |g_m|/\lambda\), \(a_0\gets x_m-r\), and \(b_0\gets x_m+r\).
\State Set \(N\gets \left\lceil \log_2(1+2r\sqrt L)\right\rceil\).
\For{\(j=0,\dots,N-1\)}
    \State Compute \(c_j\gets (a_j+b_j)/2\) and \(h_j\gets \partial_u f_m(c_j;x_{-m})\).
    \State If \(h_j\le 0\), set \((a_{j+1},b_{j+1})\gets(c_j,b_j)\); otherwise set \((a_{j+1},b_{j+1})\gets(a_j,c_j)\).
    \State Uncompute \(h_j\) and temporary arithmetic registers.
\EndFor
\State Write \(\widehat u_m(x)\gets (a_N+b_N)/2\) into the center register \(\mathrm c\).
\State Uncompute all remaining workspace registers.
\end{algorithmic}
\end{algorithm}

Given the conditional center produced by \(\textsc{QModeLocate}_m\), the next procedure coherently prepares the one-dimensional conditional sample state by combining the rejection-envelope sampler with fixed-point amplification.

\begin{algorithm}[H]
\caption{\(\textsc{QCondSample}_m\)}
\label{alg:qcondsample-clean}
\begin{algorithmic}[1]
\Statex \textbf{Input:} \(\ket{x}\ket{0}_{\mathrm c}\ket{0}_{\mathrm{out}}\ket{0}_{\mathrm w}\), precision \(\eta\).
\Statex \textbf{Output:} \(\ket{x}\ket{0}_{\mathrm c}\int \sqrt{\pi_m(u\mid x_{-m})}\ket{u}_{\mathrm{out}}\,\d u\ket{0}_{\mathrm w}\), up to error \(\eta\).

\State Run \(\textsc{QModeLocate}_m\) from \Cref{alg:q-mode-locate} to write \(\widehat u_m(x)\) into the center register.
\State Using \(\widehat u_m(x)\), instantiate the normalized local oracle for \(\bar f_m(\cdot;x)\) from \Cref{prop:q-mode-location}.
\State Reversibly run the envelope construction of \Cref{lem:1d-envelope} with this normalized oracle, producing envelope data \(\Theta_{m,x}\) in workspace.
\State Let \(\widetilde p_{m,x}(v):=e^{-\bar f_m(v;x)}\), let \(\widetilde q_{m,x}\) be the envelope described by \(\Theta_{m,x}\), set \(q_{m,x}:=\widetilde q_{m,x}/\int \widetilde q_{m,x}\), and set \(\alpha_{m,x}(v):=\widetilde p_{m,x}(v)/\widetilde q_{m,x}(v)\).
\State Coherently implement the rejection-sampling unitary
\[
\ket{0}\ket{0}
\mapsto
\int \sqrt{q_{m,x}(v)}\,\ket{v}
\Bigl(
\sqrt{\alpha_{m,x}(v)}\ket{1}
+
\sqrt{1-\alpha_{m,x}(v)}\ket{0}
\Bigr)\,\d v .
\]

\State Apply coherent rejection by fixed-point amplification to obtain the accepted normalized sample state to error \(\eta\).
\State Apply the affine map \(u=\widehat u_m(x)+v/\sqrt{\lambda}\) to the sample register.
\State Uncompute all workspace registers, including the envelope data and the center register.
\end{algorithmic}
\end{algorithm}

The final procedure combines the coherent conditional samplers into Gibbs walks and applies the slowly varying quantum-walk state-conversion routine along the Gaussian cooling schedule.
\begin{algorithm}[H]
\caption{\(\textsc{QRandomScanGibbs}(\varepsilon)\)}
\label{alg:qrsgs}
\begin{algorithmic}[1]
\State Set the conditional-sampler precision \(\eta\gets \widetilde\Theta(\varepsilon/(\sqrt{\kappa}\,d))\).
\State Construct the Gaussian cooling schedule \(\pi_0,\pi_1,\dots,\pi_{s+1}=\pi\) from \Cref{lem:qgibbs-overlap}.
\State Prepare the initial Gaussian square-root state \(\ket{\pi_0}\).
\For{\(i=0,1,\dots,s+1\)}
    \State Apply \Cref{thm:qcondsample} to the stage-\(i\) potential to obtain \(U_{m,\eta}^{\mathrm{cond},i}\) for all \(m\in[d]\).
    \State Use \Cref{lem:qgibbs-update} to assemble the implemented update isometry \(\widetilde T_{\eta,i}^{\mathrm{GS}}\).
    \State Form \(\widetilde W_{\eta,i}^{\mathrm{GS}}:=S_{\mathrm{GS}}\bigl(2\widetilde T_{\eta,i}^{\mathrm{GS}}(\widetilde T_{\eta,i}^{\mathrm{GS}})^\dagger-I\bigr)\), and provide controlled access to \(\widetilde W_{\eta,i}^{\mathrm{GS}}\) and \((\widetilde W_{\eta,i}^{\mathrm{GS}})^\dagger\).
\EndFor
\State Apply \Cref{thm:slowly-varying-qwalk} to the implemented walk sequence, starting from \(\ket{\pi_0}\), with accuracy \(\Theta(\varepsilon)\).
\State Output the resulting state \(\ket{\widetilde\pi}\).
\end{algorithmic}
\end{algorithm}

\section{Proof details}
\label{app:proof-details}

We collect the proofs of the technical statements used in the main text.
The organization follows the construction: local conditional access, coherent conditional sampling, Gibbs-walk implementation, cooling, and the final sampler.

\subsection{Local conditional access and mode location}
\label{app:proof-mode-location}

We first justify that the one-dimensional conditional potentials can be evaluated using only incident local clauses, and then prove the mode-location guarantee used to normalize the conditional target.

\begin{proposition}[Local conditional oracle access]
\label{prop:local-conditional-access}
For each \(m\in[d]\), one evaluation of \(f_m(u;x_{-m})\) can be implemented using \(O(|D_m|)\) queries to the local clause oracle \(\cO_f^{\mathrm{loc}}\).
Likewise, one evaluation of \(\partial_u f_m(u;x_{-m})\) can be implemented using \(O(|D_m|)\) queries to the local clause-gradient oracle \(\cO_{\nabla f}^{\mathrm{loc}}\).
\end{proposition}

\begin{proof}
By definition,
\[
    f_m(u;x_{-m})
    =
    \sum_{a\in D_m}
        \psi_a\bigl((u,x_{-m})_{S_a}\bigr).
\]
Thus evaluating \(f_m(u;x_{-m})\) requires querying exactly the clauses \(a\in D_m\) and summing their values by reversible arithmetic.

For the derivative, let \(\iota_a(m)\) denote the position of coordinate \(m\) inside the local coordinate list \(S_a\).
Then
\[
    \partial_u f_m(u;x_{-m})
    =
    \sum_{a\in D_m}
        \partial_{\iota_a(m)}
        \psi_a\bigl((u,x_{-m})_{S_a}\bigr).
\]
Each summand is obtained from one query to the clause-gradient oracle, followed by extracting the component corresponding to \(m\).
This gives the claimed \(O(|D_m|)\) query bounds.
\end{proof}

We now prove the mode-location primitive.

\begin{proof}[Proof of \Cref{prop:q-mode-location}]
Fix \(m\in[d]\) and \(x\in\R^d\), and write
\[
    g(u):=\partial_u f_m(u;x_{-m}).
\]
Here \(\partial_m f(x)\) denotes the \(m\)-th coordinate of the full gradient \(\nabla f(x)\).
Since \(f_m(\cdot;x_{-m})\) contains exactly the clauses of \(f\) that depend on coordinate \(m\), while all non-incident clauses are independent of \(u\), we have
\[
    g(x_m)
    =
    \partial_u f_m(x_m;x_{-m})
    =
    \partial_m f(x).
\]

Because \(f_m(\cdot;x_{-m})\) is \(\lambda\)-strongly convex, its derivative \(g\) is \(\lambda\)-strongly monotone:
\[
    (g(u)-g(v))(u-v)\ge \lambda (u-v)^2
    \qquad\text{for all }u,v\in\R .
\]
Taking \(v=u_m^\star(x_{-m})\), where \(g(v)=0\), and then setting \(u=x_m\), gives
\[
    g(x_m)\bigl(x_m-u_m^\star(x_{-m})\bigr)
    \ge
    \lambda |x_m-u_m^\star(x_{-m})|^2 .
\]
If \(x_m=u_m^\star(x_{-m})\), the desired bound is trivial.
Otherwise, dividing by \(|x_m-u_m^\star(x_{-m})|\) and using \(g(x_m)=\partial_m f(x)\) yields
\[
    |x_m-u_m^\star(x_{-m})|
    \le
    \frac{|g(x_m)|}{\lambda}
    =
    \frac{|\partial_m f(x)|}{\lambda}.
\]
Thus the initial interval \([a_0,b_0]\) constructed in \Cref{alg:q-mode-locate}, with radius \(r=|\partial_m f(x)|/\lambda\), contains the conditional minimizer.

The derivative \(g(u)=\partial_u f_m(u;x_{-m})\) is monotone, so each bisection step preserves an interval containing the unique zero of \(g\), namely \(u_m^\star(x_{-m})\).
The initial width is \(2r\), and after \(N=\lceil\log_2(1+2r\sqrt L)\rceil\) bisection steps the width is at most
\[
    \frac{2r}{2^N}
    \le
    \frac{2r}{1+2r\sqrt L}
    \le
    L^{-1/2}.
\]
Hence the midpoint \(\widehat u_m(x)\) satisfies
\[
    |\widehat u_m(x)-u_m^\star(x_{-m})|
    \le
    L^{-1/2}.
\]
Each derivative evaluation costs \(O(|D_m|)\) local clause-gradient queries by \Cref{prop:local-conditional-access}, and the number of bisection steps is
$    O\left(
        \log\!\left(1+\frac{\sqrt L}{\lambda}|\partial_m f(x)|\right)
    \right).$
This proves the query bound for mode location.

It remains to check the normalized conditional potential.
The affine change of variables \(u=\widehat u_m(x)+v/\sqrt{\lambda}\) scales both the strong-convexity and smoothness parameters by \(1/\lambda\).
Since \(f_m(\cdot;x_{-m})\) is \(\lambda\)-strongly convex and \(L\)-smooth, \(\bar f_m(\cdot;x)\) is \(1\)-strongly convex and \(\kappa=L/\lambda\)-smooth.
Its minimizer is
\[
    v_m^\star(x)
    =
    \sqrt{\lambda}\bigl(u_m^\star(x_{-m})-\widehat u_m(x)\bigr),
\]
so
\[
    |v_m^\star(x)|
    \le
    \sqrt{\lambda}L^{-1/2}
    =
    \kappa^{-1/2}.
\]

Finally, evaluating \(\bar f_m(v;x)\) requires evaluating \(f_m\) at the two local points \(\widehat u_m(x)+v/\sqrt{\lambda}\) and \(\widehat u_m(x)\), while evaluating \(\partial_v\bar f_m(v;x)\) requires evaluating \(\partial_u f_m\) at \(\widehat u_m(x)+v/\sqrt{\lambda}\), followed by multiplication by \(\lambda^{-1/2}\).
The claimed \(O(|D_m|)\) query bounds follow again from \Cref{prop:local-conditional-access}.
\end{proof}

\subsection{Coherent conditional sampling}
\label{app:proof-conditional-sampler}

We next prove that the one-dimensional rejection-envelope sampler can be implemented coherently using only local queries.

\begin{lemma}[Coherent rejection by fixed-point amplification~\cite{YLC14}]
\label{lem:coherent-rejection}
Let \(A\) be a unitary of the form
\[
A\ket{0}\ket{0}
=
\int \sqrt{q(v)}\,\ket{v}
\Bigl(
\sqrt{\alpha(v)}\ket{1}
+
\sqrt{1-\alpha(v)}\ket{0}
\Bigr)\,\d v,
\]
where \(q\) is a probability density, \(0\le \alpha(v)\le 1\), and
\(p_{\mathrm{acc}}:=\int q(v)\alpha(v)\,\d v\ge p_0\).
Then, for any \(\eta\in(0,1)\), one can prepare a state within Euclidean distance \(\eta\) of
\[
\frac{1}{\sqrt{p_{\mathrm{acc}}}}
\int \sqrt{q(v)\alpha(v)}\,\ket{v}\,\d v
\]
using \(O(p_0^{-1/2}\log(1/\eta))\) applications of \(A\), \(A^\dagger\), and elementary reflections.
\end{lemma}

In the application, \(A\) will be the coherent rejection-sampling unitary associated with the normalized conditional target \(\widetilde p_{m,x}(v):=e^{-\bar f_m(v;x)}\).
The proposal density \(q_{m,x}\) is the normalized envelope produced by the deterministic construction of \Cref{lem:1d-envelope}, and the acceptance probability is \(\alpha_{m,x}(v)=\widetilde p_{m,x}(v)/\widetilde q_{m,x}(v)\). Now, we prove \Cref{thm:qcondsample}.
\begin{proof}[Proof of \Cref{thm:qcondsample}]
Fix \(m\in[d]\) and \(x\in\R^d\).
By \Cref{prop:q-mode-location}, \(\textsc{QModeLocate}_m\) computes a center \(\widehat u_m(x)\) with
\(|\widehat u_m(x)-u_m^\star(x_{-m})|\le L^{-1/2}\).
It also gives a normalized conditional potential \(\bar f_m(\cdot;x)\) that is \(1\)-strongly convex and \(\kappa\)-smooth, with minimizer within distance \(\kappa^{-1/2}\) of the origin.
Therefore \Cref{lem:1d-envelope} applies to \(\phi=\bar f_m(\cdot;x)\).

Let \(\widetilde p_{m,x}(v):=e^{-\bar f_m(v;x)}\).
The envelope construction produces finite data \(\Theta_{m,x}\) describing an envelope \(\widetilde q_{m,x}\), its normalized proposal \(q_{m,x}:=\widetilde q_{m,x}/\int \widetilde q_{m,x}\), and the acceptance function \(\alpha_{m,x}(v):=\widetilde p_{m,x}(v)/\widetilde q_{m,x}(v)\).
The construction is deterministic, so \(\Theta_{m,x}\) can be computed reversibly from coherent oracle access to \(\bar f_m(\cdot;x)\).
By \Cref{prop:q-mode-location}, each query to this normalized oracle costs \(O(|D_m|)\) local queries, and \Cref{lem:1d-envelope} uses \(O(\log\log\kappa)\) such queries.
Thus the envelope construction costs \(O(|D_m|\log\log\kappa)\) local queries.

Conditioned on \(\Theta_{m,x}\), we implement the rejection unitary
\[
A_{m,x}\ket{0}\ket{0}
=
\int \sqrt{q_{m,x}(v)}\,\ket{v}
\Bigl(
\sqrt{\alpha_{m,x}(v)}\ket{1}
+
\sqrt{1-\alpha_{m,x}(v)}\ket{0}
\Bigr)\,\d v .
\]
Since \(\int \widetilde q_{m,x}\le C\int \widetilde p_{m,x}\), the acceptance probability is at least \(1/C=\Omega(1)\).
By \Cref{lem:coherent-rejection}, fixed-point amplification prepares the accepted branch to Euclidean error at most \(\eta\) using \(O(\log(1/\eta))\) applications of \(A_{m,x}\) and \(A_{m,x}^{\dagger}\).
Each application evaluates the normalized target at the proposed point and hence costs \(O(|D_m|)\) local clause-oracle queries.
This contributes \(O(|D_m|\log(1/\eta))\) further local queries.

After conditioning on acceptance, the \(v\)-register is the square-root encoding of the density proportional to \(q_{m,x}(v)\alpha_{m,x}(v)\), which is the same density as \(\widetilde p_{m,x}(v)\) after normalization.
Applying the affine map \(u=\widehat u_m(x)+v/\sqrt{\lambda}\) therefore gives the square-root encoding of the original conditional law \(\pi_m(\cdot\mid x_{-m})\).
All auxiliary data, including \(\Theta_{m,x}\) and the center register, are deterministic functions of \(x\) and are uncomputed reversibly.

Combining the mode-location, envelope-construction, and coherent-rejection costs gives
\[
O\!\left(
|D_m|
\left[
\log\!\left(1+\frac{\sqrt L}{\lambda}|\partial_m f(x)|\right)
+
\log\log\kappa
+
\log(1/\eta)
\right]
\right)
\]
local queries.
By \Cref{lem:gradient-overhead-truncation} and \Cref{rem:gradient-overhead}, the gradient-dependent logarithm is polylogarithmic after the standard finite-precision truncation of the continuous register.
Thus the cost is \(\widetilde O(|D_m|)\).
\end{proof}

\begin{lemma}[Truncation controls the mode-location overhead]
\label{lem:gradient-overhead-truncation}
Let \(f:\R^d\to\R\) be \(L\)-smooth and \(\lambda\)-strongly convex, let
\(\pi(x)\propto e^{-f(x)}\), and let \(x^\star\) be the minimizer of \(f\).
For any \(\eta_{\rm tr}\in(0,1)\), define
$
    K_R:=\{x\in\R^d:\|x-x^\star\|\le R\}$, where
$R:=\sqrt{d/\parens{\lambda\eta_{\rm tr}}}$.
Then \(\pi(K_R)\ge 1-\eta_{\rm tr}\). Moreover, for every \(x\in K_R\) and every \(m\in[d]\),
$\log\!\left(1+\frac{\sqrt L}{\lambda}|\partial_m f(x)|\right)
\le
\log\!\left(1+\kappa^{3/2}\sqrt{\frac d{\eta_{\rm tr}}}\right)$.
In particular, if \(\eta_{\rm tr}^{-1}\) is polynomial in the problem parameters, this overhead is polylogarithmic.
\end{lemma}

\begin{proof}
We first bound the mass outside \(K_R\).
By integration by parts,
\[
    \E_\pi\!\left[\langle \nabla f(X),X-x^\star\rangle\right]=d.
\]
Since \(f\) is \(\lambda\)-strongly convex and \(\nabla f(x^\star)=0\), we have
\(\langle \nabla f(x),x-x^\star\rangle\ge \lambda\|x-x^\star\|^2\).
Therefore \(\E_\pi\|X-x^\star\|^2\le d/\lambda\).
Markov's inequality gives
\[
    \pi(K_R^c)
    =
    \Pr(\|X-x^\star\|>R)
    \le
    \frac{\E_\pi\|X-x^\star\|^2}{R^2}
    \le
    \eta_{\rm tr}.
\]

Now fix \(x\in K_R\).
By \(L\)-smoothness and \(\nabla f(x^\star)=0\),
\[
    |\partial_m f(x)|
    \le
    \|\nabla f(x)\|
    =
    \|\nabla f(x)-\nabla f(x^\star)\|
    \le
    L\|x-x^\star\|
    \le
    LR.
\]
Hence
\[
    \frac{\sqrt L}{\lambda}|\partial_m f(x)|
    \le
    \frac{L^{3/2}}{\lambda}
    \sqrt{\frac{d}{\lambda\eta_{\rm tr}}}
    =
    \kappa^{3/2}\sqrt{\frac d{\eta_{\rm tr}}},
\]
which proves the claim.
\end{proof}

\begin{remark}[Mode-location overhead in finite precision]
\label{rem:gradient-overhead}
In a finite-precision implementation, the continuous register is first truncated to a large ball \(K_R\) and then discretized on that ball.
The preceding lemma shows that, on such a truncation domain with \(\pi(K_R)\ge 1-\eta_{\rm tr}\), the gradient-dependent factor in \Cref{thm:qcondsample},
$    \log\!\left(1+\frac{\sqrt L}{\lambda}|\partial_m f(x)|\right)$,
is polylogarithmic whenever \(\eta_{\rm tr}^{-1}\) is polynomial in \(d,\kappa,1/\varepsilon\).
The same argument applies to the regularized potentials used in the Gaussian cooling schedule, with their corresponding smoothness and strong-convexity parameters.
Thus this factor is suppressed in the \(\widetilde O(|D_m|)\) local query bound.
\end{remark}

\subsection{Gibbs update isometry and phase gap}
\label{app:proof-gibbs-walk}

We now prove the two statements that connect the coherent conditional sampler to the Szegedy walk: the implementation of the update isometry and the phase-gap bound.

\subsubsection{Proof of the update isometry}
\label{app:proof-gibbs-isometry}

\begin{proof}[Proof of \Cref{lem:qgibbs-update}]
Let \(\mathcal E_m\) denote the isometric embedding that writes the scan label \(m\) into the proposal register.
Thus the ranges of \(\mathcal E_m\) and \(\mathcal E_{m'}\) are orthogonal whenever \(m\neq m'\).

For each coordinate \(m\), let \(B_{m,\star}\) be the ideal branch isometry that maps
\[
    \ket{x}
    \longmapsto
    \ket{x}
    \int \sqrt{\pi_m(u\mid x_{-m})}\,
    \ket{x^{(m\leftarrow u)}}\,\d u ,
\]
and let \(\widetilde B_{m,\eta}\) be the corresponding implemented branch obtained by using \(U_{m,\eta}^{\mathrm{cond}}\).
With this notation,
\[
    T_\star^{\mathrm{GS}}
    =
    \frac1{\sqrt d}
    \sum_{m=1}^d
    \mathcal E_m B_{m,\star},
    \qquad
    \widetilde T_{\eta}^{\mathrm{GS}}
    =
    \frac1{\sqrt d}
    \sum_{m=1}^d
    \mathcal E_m \widetilde B_{m,\eta}.
\]
The assumption on the conditional samplers gives
\(\|\widetilde B_{m,\eta}-B_{m,\star}\|_{\mathrm{op}}\le \eta\) for every \(m\).

For any input state \(\ket{\varphi}\), orthogonality of the scan-label subspaces gives
\[
\begin{aligned}
\left\|
(\widetilde T_{\eta}^{\mathrm{GS}}-T_\star^{\mathrm{GS}})
\ket{\varphi}
\right\|_2^2
&=
\frac1d
\sum_{m=1}^d
\left\|
\mathcal E_m
(\widetilde B_{m,\eta}-B_{m,\star})
\ket{\varphi}
\right\|_2^2  \\
&=
\frac1d
\sum_{m=1}^d
\left\|
(\widetilde B_{m,\eta}-B_{m,\star})
\ket{\varphi}
\right\|_2^2
\le
\eta^2 \|\ket{\varphi}\|_2^2 .
\end{aligned}
\]
Taking the supremum over unit vectors \(\ket{\varphi}\) yields
\(\|\widetilde T_{\eta}^{\mathrm{GS}}-T_\star^{\mathrm{GS}}\|_{\mathrm{op}}\le \eta\).

It remains to count queries.
The branch for coordinate \(m\) uses the conditional sampler \(U_{m,\eta}^{\mathrm{cond}}\), whose local query cost is \(\widetilde O(|D_m|)\).
The controlled implementation can pad all branches to the worst coordinate cost, giving one application of \(\widetilde T_{\eta}^{\mathrm{GS}}\) using \(\widetilde O(\max_m |D_m|)=\widetilde O(\Delta)\) local queries.
The same bound applies to \(\widetilde T_{\eta}^{\mathrm{GS}}{}^\dagger\) by reversing the circuit.
\end{proof}

\subsubsection{Proof of the phase-gap bound}
\label{app:proof-gibbs-spectrum}

\begin{proof}[Proof of \Cref{lem:qgibbs-spectrum}]
For each coordinate \(m\), define the ideal coordinate-update isometry
\[
T_m\ket{x}
=
\ket{x}\ket m
\int \sqrt{\pi_m(u\mid x_{-m})}\,
\ket{x^{(m\leftarrow u)}}\,\d u .
\]
Then \(T_\star^{\mathrm{GS}}=d^{-1/2}\sum_{m=1}^d T_m\).
The scan labels are orthogonal, so the discriminant
\(D_\star^{\mathrm{GS}}:=(T_\star^{\mathrm{GS}})^\dagger S_{\mathrm{GS}}T_\star^{\mathrm{GS}}\) decomposes as
\[
D_\star^{\mathrm{GS}}
=
\frac1d\sum_{m=1}^d T_m^\dagger S_{\mathrm{GS}}T_m .
\]
We identify each summand.
Let \(P_m\) be the coordinate-\(m\) Gibbs update.
For this update, \(P_m(x,\d y)\) is supported on states \(y\) with \(y_{-m}=x_{-m}\), and its one-dimensional transition density is \(\pi_m(y_m\mid x_{-m})\).
Thus \(T_m^\dagger S_{\mathrm{GS}}T_m\) has kernel
\[
\mathbf 1\{x_{-m}=y_{-m}\}
\sqrt{\pi_m(x_m\mid x_{-m})\pi_m(y_m\mid x_{-m})}.
\]
Equivalently, if \(V:L^2(\pi)\to L^2(\d x)\) is the unitary \(Vg=\sqrt{\pi}\,g\), then
\[
T_m^\dagger S_{\mathrm{GS}}T_m
=
V P_m V^{-1}.
\]
Averaging over \(m\) gives
\[
D_\star^{\mathrm{GS}}
=
V\left(\frac1d\sum_{m=1}^d P_m\right)V^{-1}
=
V P^{\mathrm{GS}} V^{-1}.
\]
Hence the discriminant has the same spectrum as the classical random-scan Gibbs kernel.

By the spectral mapping theorem for Szegedy walks~\cite[Theorem~3.1]{CCH23}, an eigenvalue \(\lambda\in[-1,1]\) of the discriminant gives eigenphases \(\pm\arccos(\lambda)\) of \(W_\star^{\mathrm{GS}}\).
The stationary eigenvalue \(\lambda=1\) gives the stationary eigenphase \(0\).
On the orthogonal complement of the stationary state, \(\lambda\le 1-\gap(P^{\mathrm{GS}})\), and therefore the smallest nonzero eigenphase is at least
\[
\arccos(1-\gap(P^{\mathrm{GS}}))
=
\Omega(\sqrt{\gap(P^{\mathrm{GS}})}).
\]
Using \Cref{lem:gibbs-classical-gap}, we obtain \(\gamma_{\mathrm{GS}}=\Omega(1/\sqrt{\kappa d})\).

It remains to compare the implemented and ideal walks.
Let \(\widetilde\Pi_{\eta}^{\mathrm{GS}}:=\widetilde T_{\eta}^{\mathrm{GS}}(\widetilde T_{\eta}^{\mathrm{GS}})^\dagger\).
If \(\|\widetilde T_{\eta}^{\mathrm{GS}}-T_\star^{\mathrm{GS}}\|_{\mathrm{op}}\le \eta\), then
\[
\|\widetilde\Pi_{\eta}^{\mathrm{GS}}-\Pi_\star^{\mathrm{GS}}\|_{\mathrm{op}}
\le
\|(\widetilde T_{\eta}^{\mathrm{GS}}-T_\star^{\mathrm{GS}})(\widetilde T_{\eta}^{\mathrm{GS}})^\dagger\|_{\mathrm{op}}
+
\|T_\star^{\mathrm{GS}}((\widetilde T_{\eta}^{\mathrm{GS}})^\dagger-(T_\star^{\mathrm{GS}})^\dagger)\|_{\mathrm{op}}
\le 2\eta .
\]
Therefore the corresponding reflections differ by at most \(4\eta\) in operator norm.
Multiplication by the same swap \(S_{\mathrm{GS}}\) preserves the norm, which proves \(\|\widetilde W_{\eta}^{\mathrm{GS}}-W_\star^{\mathrm{GS}}\|_{\mathrm{op}}=O(\eta)\).
\end{proof}

\subsection{Cooling schedule and final sampler}
\label{app:proof-main-sampler}

We analyze the Gaussian cooling schedule, recall the slowly varying quantum-walk state-conversion theorem, and then prove the main sampler and the warm-start refinement.

\subsubsection{Gaussian cooling estimates}
\label{app:proof-cooling}
\begin{proof}[Proof of \Cref{lem:qgibbs-overlap}]
As in the main text, translate coordinates so that \(x^\star=0\) and subtract the constant \(f(x^\star)\), so that \(f(0)=0\).
For \(1\le i\le s\), write
\[
    a_i:=\frac{1}{2\sigma_i^2},
    \qquad
    \pi_i(x)\propto e^{-f(x)-a_i\|x\|^2},
\]
and set \(a_{s+1}:=0\), so that \(\pi_{s+1}=\pi\).
Since \(\sigma_{i+1}^2=(1+1/\sqrt d)\sigma_i^2\), we have
\(a_{i+1}=a_i/(1+1/\sqrt d)\).

We first bound the overlap between \(\pi_0\) and \(\pi_1\).
The initial distribution is \(\pi_0=N(0,\sigma_1^2 I)\), with \(\sigma_1^2=1/(2Ld)\).
Since \(f(0)=0\) and \(f\) is \(L\)-smooth, \(0\le f(x)\le (L/2)\|x\|^2\).
Thus, for \(X\sim\pi_0\),
\[
|\langle \pi_0\mid \pi_1\rangle|
=
\frac{\E_{\pi_0}[e^{-f(X)/2}]}
     {\E_{\pi_0}[e^{-f(X)}]^{1/2}} .
\]
The denominator is at most one, while
\[
\E_{\pi_0}[e^{-f(X)/2}]
\ge
\E_{\pi_0}\!\left[e^{-L\|X\|^2/4}\right]
=
\left(1+\frac{1}{4d}\right)^{-d/2}
=
\Omega(1).
\]
Hence \(|\langle \pi_0\mid \pi_1\rangle|=\Omega(1)\).

We next control the overlaps among the regularized targets.
For \(a\ge 0\), define
\[
    Z(a):=\int_{\R^d} e^{-f(x)-a\|x\|^2}\,\d x,
    \qquad
    \Phi(a):=\log Z(a),
\]
and let \(\nu_a\) denote the probability measure with density proportional to
\(e^{-f(x)-a\|x\|^2}\).
For two parameters \(a,b\ge 0\), the square-root overlap is
\[
    |\langle \nu_a\mid\nu_b\rangle|
    =
    \frac{Z((a+b)/2)}{\sqrt{Z(a)Z(b)}}.
\]
Equivalently, defining
\[
    \Delta(a,b)
    :=
    \frac{\Phi(a)+\Phi(b)}2-\Phi\!\left(\frac{a+b}{2}\right),
\]
we have
$|\langle \nu_a\mid\nu_b\rangle|=e^{-\Delta(a,b)}$.

We now bound this midpoint convexity gap.
A direct differentiation gives
\(\Phi'(a)=-\E_a\|X\|^2\) and
\(\Phi''(a)=\operatorname{Var}_a(\|X\|^2)\), where \(X\sim\nu_a\).
The potential \(f(x)+a\|x\|^2\) is \((\lambda+2a)\)-strongly convex.
By the Poincar\'e inequality for strongly log-concave measures,
\[
    \Phi''(a)
    =
    \operatorname{Var}_a(\|X\|^2)
    \le
    \frac{1}{\lambda+2a}\,
    \E_a\|\nabla\|X\|^2\|^2.
\]
Since \(\nabla\|x\|^2=2x\) and \(\E_a\|X\|^2\le d/(\lambda+2a)\), we obtain
\[
    \Phi''(a)
    \le
    \frac{4d}{(\lambda+2a)^2}.
\]

For \(1\le i\le s-1\), take \(a=a_i\), \(b=a_{i+1}\), and let \(\delta:=1/\sqrt d\).
Then \(a=(1+\delta)b\) and \(a-b=\delta b\).
Taylor's theorem gives
\[
\Delta(a,b)
\le
\frac{(a-b)^2}{8}\sup_{t\in[b,a]}\Phi''(t)
\le
\frac{\delta^2b^2}{8}\cdot
\frac{4d}{(\lambda+2b)^2}
\le
\frac18,
\]
where the last step uses \(\delta^2 d=1\) and \(\lambda+2b\ge 2b\).
Therefore every intermediate overlap satisfies
\[
    |\langle \pi_i\mid\pi_{i+1}\rangle|
    =
    e^{-\Delta(a_i,a_{i+1})}
    \ge e^{-1/8}
    =
    \Omega(1).
\]

It remains to handle the final transition from \(\pi_s\) to \(\pi\).
Here \(a=a_s\) and \(b=0\).
By the choice \(\sigma_s^2=C d/\lambda\), we have \(a_s=\lambda/(2Cd)\).
Using the same Taylor bound,
\[
\Delta(a_s,0)
\le
\frac{a_s^2}{8}\sup_{t\in[0,a_s]}\Phi''(t)
\le
\frac{a_s^2}{8}\cdot \frac{4d}{\lambda^2}
=
O\!\left(\frac{1}{C^2d}\right).
\]
Consequently,
\[
    |\langle \pi_s\mid\pi\rangle|
    =
    e^{-\Delta(a_s,0)}
    \ge
    \exp\!\left(-O\!\left(\frac{1}{C^2d}\right)\right)
    =
    \Omega(1).
\]
This proves the constant-overlap claim for all consecutive stages.

The number of stages is determined by increasing \(\sigma_i^2\) by the factor \(1+1/\sqrt d\) from \(\sigma_1^2=1/(2Ld)\) until \(\sigma_s^2=C d/\lambda\).
Equivalently, \(a_i\) decreases from \(a_1=Ld\) to \(a_s=\Theta(\lambda/d)\), so
\[
    s
    =
    O\!\left(\sqrt d\log\frac{a_1}{a_s}\right)
    =
    O\!\left(\sqrt d\log(\kappa d)\right)
    =
    \widetilde O(\sqrt d).
\]

Finally, we prove the spectral and phase-gap bounds.
For \(1\le i\le s\), the stage-\(i\) potential is
\(f_i(x)=f(x)+a_i\|x\|^2\).
Its smoothness and strong-convexity parameters are \(L+2a_i\) and \(\lambda+2a_i\), respectively, so its condition number is at most \(\kappa=L/\lambda\).
The initial Gaussian stage has condition number \(1\), and the final stage has condition number \(\kappa\).
Applying \Cref{lem:gibbs-classical-gap} to each stage gives
\(\gap(P_i^{\mathrm{GS}})=\Omega(1/(\kappa d))\).
The phase-gap bound
\(\gamma_{\mathrm{GS},i}=\Omega(1/\sqrt{\kappa d})\)
then follows from \Cref{lem:qgibbs-spectrum}.
\end{proof}

\subsubsection{Slowly varying quantum walks}
\label{app:slowly-varying-qwalk}

We use the following standard quantum state-conversion theorem for slowly varying Markov chains.

\begin{theorem}[Quantum speedup for slowly varying Markov chains~{\cite[Theorem~2]{WA08}}]
\label{thm:slowly-varying-qwalk}
Let \(M_0,\dots,M_t\) be reversible Markov chains with stationary distributions \(\rho_0,\dots,\rho_t\), and let \(\ket{\rho_i}\) denote the square-root state of \(\rho_i\).
Suppose that the adjacent square-root states satisfy \(|\langle \rho_i\mid \rho_{i+1}\rangle|\ge p\) for all \(i=0,\dots,t-1\), and that the Szegedy walk associated with each \(M_i\) has phase gap at least \(\gamma\).
Given \(\ket{\rho_0}\), for any \(\varepsilon\in(0,1)\), there is a quantum algorithm that prepares a state \(\varepsilon\)-close in Euclidean norm to \(\ket{\rho_t}\) using \(\widetilde O(t/(p\gamma))\) controlled applications of the quantum walk operators associated with the \(M_i\)'s and their inverses.
\end{theorem}

\subsubsection{Proof of the main sampler}
\label{app:proof-main-sampler-detail}

\begin{proof}[Proof of \Cref{thm:q-rsgs-main}]
By \Cref{lem:qgibbs-overlap}, the Gaussian cooling schedule has \(s+1=\widetilde O(\sqrt d)\) transitions, adjacent square-root states have constant overlap, and every stagewise Gibbs walk has phase gap \(\gamma_{\mathrm{GS},i}=\Omega(1/\sqrt{\kappa d})\).
Applying \Cref{thm:slowly-varying-qwalk} with \(t=s+1=\widetilde O(\sqrt d)\), \(p=\Omega(1)\), and \(\gamma=\Omega(1/\sqrt{\kappa d})\), the ideal state-conversion procedure uses \(\widetilde O(\sqrt d\cdot\sqrt{\kappa d})=\widetilde O(\sqrt{\kappa}\,d)\) controlled applications of the stagewise Gibbs walks.

We next bound the cost of one implemented walk application.
At stage \(i\), the potential is either the initial Gaussian potential, a regularized potential \(f_i(x)=f(x)+\|x\|^2/(2\sigma_i^2)\), or the final potential \(f\).
The added Gaussian term is coordinate-separable, so it adds only unary clauses and increases each occurrence number by at most one.
Hence, by \Cref{thm:qcondsample}, the coordinate-\(m\) conditional sampler at any stage has local query cost \(\widetilde O(|D_m|+1)\).
By \Cref{lem:qgibbs-update}, the corresponding update isometry \(\widetilde T_{\eta,i}^{\mathrm{GS}}\) costs \(\widetilde O(\Delta+1)\) local queries.
The walk \(\widetilde W_{\eta,i}^{\mathrm{GS}}=S_{\mathrm{GS}}(2\widetilde T_{\eta,i}^{\mathrm{GS}}(\widetilde T_{\eta,i}^{\mathrm{GS}})^\dagger-I)\) uses a constant number of calls to \(\widetilde T_{\eta,i}^{\mathrm{GS}}\) and its adjoint, plus query-free reversible operations.
Thus one controlled application of \(\widetilde W_{\eta,i}^{\mathrm{GS}}\), or its adjoint, costs \(\widetilde O(\Delta+1)\) local queries, which we write as \(\widetilde O(\Delta)\) for \(\Delta\ge 1\).
Multiplying by the number of walk applications gives total local query complexity \(\widetilde O(\Delta\sqrt{\kappa}\,d)\).

It remains to control the error from using implemented walks.
Let \(N=\widetilde O(\sqrt{\kappa}\,d)\) be the total number of controlled walk applications.
For every stage \(i\), \Cref{lem:qgibbs-update} gives \(\|\widetilde T_{\eta,i}^{\mathrm{GS}}-T_{\star,i}^{\mathrm{GS}}\|_{\mathrm{op}}\le \eta\), and \Cref{lem:qgibbs-spectrum} gives \(\|\widetilde W_{\eta,i}^{\mathrm{GS}}-W_{\star,i}^{\mathrm{GS}}\|_{\mathrm{op}}=O(\eta)\), with the same bound for the adjoints.
A standard hybrid argument over the \(N\) controlled walk applications bounds the accumulated implementation error by \(O(N\eta)\).
With \(\eta=\widetilde\Theta(\varepsilon/(\sqrt{\kappa}\,d))\), this contribution is \(O(\varepsilon)\).

The slowly varying procedure itself is run with final accuracy \(\Theta(\varepsilon)\).
Choosing the hidden constants in this accuracy and in \(\eta\) sufficiently small, the state-preparation error and the implementation error together give \(\|\ket{\widetilde\pi}-\ket{\pi}\|_2\le \varepsilon\).
Finally, if \(q\) is the distribution obtained by measuring \(\ket{\widetilde\pi}\), then \(\TV(q,\pi)\le \|\ket{\widetilde\pi}-\ket{\pi}\|_2\), since total variation between the squared amplitudes is bounded by the Euclidean distance between the corresponding nonnegative square-root states.
This proves the TV guarantee.
\end{proof}

\subsubsection{Warm start}
\label{app:proof-warm-start}
The cooling schedule is unnecessary when a warm initial quantum sample state is already available.

\begin{corollary}[Warm start]
\label{cor:warm-start}
Suppose we are given a square-root state
\(\ket{\mu_0}:=\int \sqrt{\mu_0(x)}\ket{x}\,\d x\), where \(\mu_0\) is \(\beta\)-warm with respect to \(\pi\), i.e.,
\(\mathrm d\mu_0/\mathrm d\pi\le \beta\).
Then the target square-root state \(\ket{\pi}\) can be prepared to Euclidean error \(\varepsilon\) using
\(\widetilde O(\Delta\sqrt{\beta\kappa d})\)
local queries.
In particular, for a constant-warm start, the complexity is
\(\widetilde O(\Delta\sqrt{\kappa d})\).
\end{corollary}

\begin{proof}[Proof of \Cref{cor:warm-start}]
The warmness condition implies a lower bound on the initial overlap:
\[
\langle \mu_0|\pi\rangle
=
\int \sqrt{\frac{\mathrm d\mu_0}{\mathrm d\pi}}\,\mathrm d\pi
\ge
\frac1{\sqrt{\beta}}.
\]
By \Cref{lem:qgibbs-spectrum}, the Gibbs quantum walk for the target has phase gap
\(\gamma_{\mathrm{GS}}=\Omega(1/\sqrt{\kappa d})\).
Applying \Cref{thm:slowly-varying-qwalk} with no cooling schedule, initial overlap \(p\ge \beta^{-1/2}\), and phase gap \(\gamma_{\mathrm{GS}}\), we can prepare \(\ket{\pi}\) using
\(\widetilde O(p^{-1}\gamma_{\mathrm{GS}}^{-1})=\widetilde O(\sqrt{\beta\kappa d})\)
walk applications.
Each walk application costs \(\widetilde O(\Delta)\) local queries by \Cref{lem:qgibbs-update}.
The implementation error is handled as in the proof of \Cref{thm:q-rsgs-main}, by choosing the conditional-sampler precision inversely proportional to the total number of walk applications.
This gives the stated \(\widetilde O(\Delta\sqrt{\beta\kappa d})\) local query complexity.
\end{proof}

\end{document}